\documentclass[letterpaper, 10 pt, conference]{ieeeconf}  

\IEEEoverridecommandlockouts                              

\usepackage{graphics} 
\usepackage{epsfig} 
\usepackage{amssymb}  
\usepackage{xcolor}

\usepackage{caption}    

\usepackage{float}

\usepackage{amsmath}   
\usepackage{array}     
\usepackage{booktabs}  
\usepackage{multirow}  
\usepackage{makecell}  
\usepackage{graphicx}  
\usepackage{adjustbox} 

\newcolumntype{P}[1]{>{\centering\arraybackslash}p{#1}}
\newcolumntype{M}[1]{>{\centering\arraybackslash}m{#1}} 

\usepackage[ruled]{algorithm2e}
\RestyleAlgo{ruled}
\SetKwComment{Comment}{/* }{ */}

\newtheorem{proposition}{Proposition}

\title{\LARGE \bf
Learning Higher Order DC-DC Converter Control from Inversion

}

\author{
\authorblockN{Kamakshi Tatkare\authorrefmark{1}, Ufuk Topcu\authorrefmark{2}, and Brian Johnson\authorrefmark{1}}
\authorblockA{\authorrefmark{1}Chandra Family Department of Electrical and Computer Engineering,\\
}
\authorblockA{\authorrefmark{2}Department of Aerospace Engineering and Engineering Mechanics, \\}
\authorblockA{The University of Texas at Austin. \\
Email: kamakshi@utexas.edu, utopcu@utexas.edu, b.johnson@utexas.edu}
}

\begin{document}

\maketitle
\thispagestyle{empty}
\pagestyle{empty}

\begin{abstract}
Control design for high-order converters with nonlinear dynamics is generally difficult. Nonlinear controls introduce formidable model complexity whereas linear controls suffer from poor disturbance rejection. 
Our objective in this paper is to develop a simple neural-network-based controller that learns the control law from supervised trajectories generated by a model-based bounded inversion framework. First, we establish strong invertibility for the converter models. Next, we prescribe the desired closed-loop behavior through a trajectory model that defines the target dynamics. We then invert the converter dynamics numerically to obtain the corresponding duty-cycle and state trajectories. We use these model-based trajectories as supervision and train a neural network to approximate the resulting state-to-duty control law. The results show that the learned controllers accurately recover the inversion-based control action and achieve effective regulation under load and input-voltage disturbances. We validate the method on fourth order \'Cuk, SEPIC, and zeta converters.
\end{abstract}


\section{INTRODUCTION}

Regulating the output voltage of a dc-dc converter in the presence of load and input-voltage variations remains a central control problem in power electronics. Since most converters are inherently nonlinear and often operate over wide ranges of conditions, large-signal control methods attempt to obtain regulation and disturbance rejection beyond what is typically available from local small-signal designs \cite{banerjee1999nonlinear}. For basic converter topologies, several nonlinear control frameworks provide systematic ways to analyze the plant dynamics and construct feedback laws with improved large-signal behavior.

That analytical convenience weakens, however, as converter order increase. Higher order topologies such as the \'Cuk, SEPIC, and zeta converters are fourth order and more difficult to handle within a simple framework. The broader literature on these converters reflects this difficulty. For the \'Cuk converter, prior work includes cascade linearization \cite{Cuk_cascade_partialstates}, nonlinear PI regulation \cite{non_linear_PI_cuk}, sliding-mode variants  \cite{cuk_1995_smc,Cuk_IndirectSMC, Investigate_SMC}, passivity-based with online estimation \cite{PBC_cuk}, hybrid model predictive \cite{mpc_cuk}, hysteresis switching \cite{hysteresis_cuk}, extended-state-observer-based \cite{ESO_cuk}, linear quadratic regulation \cite{LQR_cuk}, and nonlinear \(H_\infty\) \cite{Hinf_cuk} control. The following have been applied to the zeta converter: PI \cite{PI_zeta2,PI_zeta}, sliding-mode and passive-output-feedback designs \cite{sira2006control,Investigate_SMC}, and model-reference adaptive control schemes \cite{adaptive1_zeta} have been reported. Finally, literature on SEPIC converter includes small-signal peak-current-controlled \cite{smallsignal_sepic_cuk,stability_sepic}, PI  \cite{PI_sepic}, sliding-mode-based approaches \cite{smc_indirect_sepic, highf_smc_sepic}, passivity- and backstepping-based \cite{passivity_backstep_sepic}, fuzzy logic control \cite{Fuzzy_sepic}, and model predictive control \cite{mpc_sepic}.

\begin{figure}[t]
  \centering
    \includegraphics{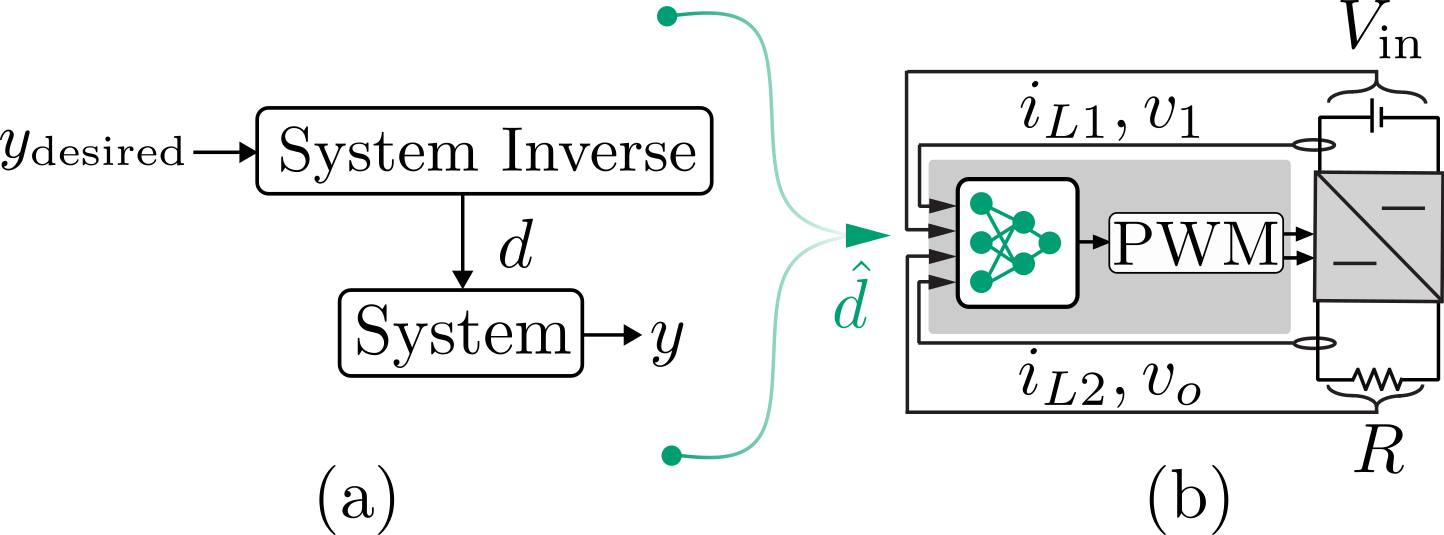}
   \caption{Inversion-based neural-network control: (a) given $y_{\mathrm{desired}}$, the system inverse computes an input $d$ that, when applied through the actuation stage, drives the converter output $y$ to follow $y_{\mathrm{desired}}$ for an operating point, and (b) a neural network maps converter states to the duty, $\hat{d}$, for regulation under input-voltage and load variations.}
  \label{fig:system_inverse}
\end{figure}

Taken together, these results show that voltage regulation for fourth-order converters is achievable, but often with important tradeoffs. Certain controllers are formulated for load disturbances only or around restricted operating conditions \cite{Cuk_cascade_partialstates,LQR_cuk}. Others are designed for third order converters \cite{non_linear_PI_cuk}. Some achieve robustness to line and load variations but rely on auxiliary loops, observers, indirect constructions, or sliding-surface regulation \cite{Cuk_IndirectSMC,ESO_cuk,Investigate_SMC,smc_indirect_sepic}. Other methods do not explicitly shape the output transient, may exhibit chattering, or introduce substantial synthesis and implementation complexity through real-time optimization or Hamilton-Jacobi-type conditions \cite{Cuk_PI_SMC,mpc_cuk,mpc_sepic,Hinf_cuk}. These recurring tradeoffs motivate the search for a control framework that remains simple to design, robust to multiple disturbance types, and does not require a converter-specific closed-form derivation.

We propose a simple solution to prescribe the closed-loop trajectory we want and then recover the duty-ratio signal by numerically inverting the converter dynamics. As shown in Fig.~\ref{fig:system_inverse}(a), for a given operating condition and desired response, the system inverse computes the control input that drives the converter along the target trajectory \cite{devasia_1996}. Repeating this procedure over many operating points and disturbance scenarios produces labeled state-input trajectories that encode the desired closed-loop behavior.

These trajectories are then used as supervised training data. Fig.~\ref{fig:system_inverse}(b) shows the resulting controller architecture. We train a multilayer perceptron (MLP) to approximate the mapping from measured converter variables to the duty ratio computed by inversion. The neural network therefore does not search over an unrestricted family of control policies. Instead, it learns a structured state-to-duty mapping induced by model-based trajectory generation. In this way, the proposed framework preserves the design logic of trajectory-based nonlinear control while avoiding the need for a topology-specific closed-form synthesis.

In this paper, we apply this framework to output voltage regulation of fourth-order dc-dc converters under load and input-voltage disturbances. We focus on the \'Cuk, SEPIC, and zeta converters and show that inversion-based trajectory generation combined with supervised learning yields simple neural-network controllers for disturbance rejection.

\begin{figure*}[t]
  \centering
\includegraphics[width=1\textwidth]{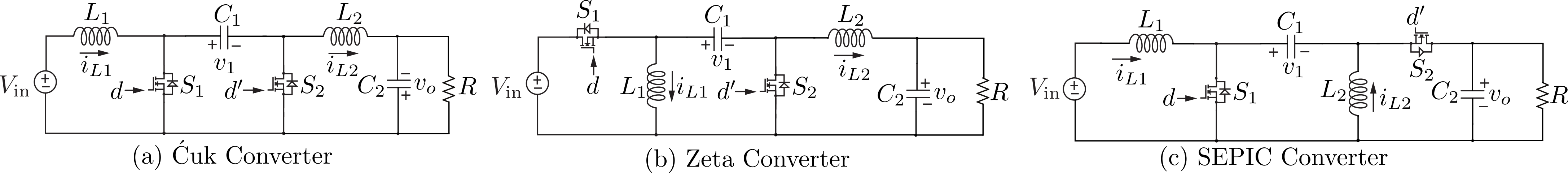}
  \caption{Power stages of ideal basic dc-dc converter.}
  \label{fig:power_stage}
\end{figure*}
 
\section{Inversion}
We choose the inversion output based on the phase properties of the corresponding input-output map. Minimum-phase systems admit bounded causal inverse trajectories, whereas nonminimum-phase systems require noncausal preimages associated with unstable zero dynamics~\cite{devasia_1996}. For the averaged linearized \'Cuk, zeta, and SEPIC converters, the small-signal output voltage is a nonminimum-phase state, while the small-signal input current is a minimum-phase state~\cite{sira2006control}. This property, which appears in many classical dc-dc converters, motivates an indirect feedback strategy: we regulate the input inductor current to its desired equilibrium and rely on the asymptotically stable zero dynamics to govern the internal dynamics and the actual output behavior~\cite{sira2006control}. Accordingly, the input inductor current will be chosen as the regulated state from here forward. 

We establish strong invertibility of the averaged converters in Fig.~\ref{fig:power_stage} using Hirschorn's framework for nonlinear control systems \cite{hirschron_1979}. Consider the single-input input-affine system
\begin{equation}
    \dot x = A(x) + uB(x), \qquad y = c(x),
    \label{eq:input_affine}
\end{equation}
for which, using chain rule,
\begin{equation}
    \dot y = L_A c(x) + u\,L_B c(x),
    \label{eq:output_derivative}
\end{equation}
where compact forms are \(L_A c(x) := \frac{\partial c}{\partial x}A(x)\) and \(L_B c(x) := \frac{\partial c}{\partial x}B(x)\). If \(L_B c(x) \neq 0\) on an open set \(\mathcal{M}\), then the system has relative order \(\alpha=1\), which is the highest order derivative of $y$, on \(\mathcal{M}\) \cite{silverman2003inversion}. As $\alpha < \infty$,  following Hirschorn's Theorem 2, the system is strongly invertible on \(\mathcal{M}\), and the input is recovered from \eqref{eq:output_derivative} as
\begin{equation}
    u = \frac{\dot y - L_A c(x)}{L_B c(x)}.
    \label{eq:inverse_general}
\end{equation}
We apply this criterion with $x$ as converter states, the regulated state as \(y = i_{L_1}\), and control effort as $u= d$.

\subsection{\'Cuk Converter}

\begin{proposition}
Consider the averaged \'Cuk converter
\begin{align*}
\dot y &= -\frac{1-d}{L_1}v_1+\frac{V_{\mathrm{in}}}{L_1},\\
\dot v_1 &= \frac{1-d}{C_1}y+\frac{d}{C_1}i_2,\\
\dot i_2 &= -\frac{d}{L_2}v_1-\frac{v_o}{L_2},\\
\dot v_o &= \frac{i_2}{C_2}-\frac{v_o}{RC_2},
\end{align*}
with state \(x=[y\;\; v_1\;\; i_2\;\; v_o]^\top\) and output \(c(x)=y=i_{L_1}\). Then the system is strongly invertible on
\[
\mathcal{M}_{\mathrm{C}}=\{x\in\mathbb{R}^4 : v_1\neq 0\}.
\]
\end{proposition}

\begin{proof}
The system can be written in the form \eqref{eq:input_affine} with
\[
A(x)=
\begin{bmatrix}
-\dfrac{v_1}{L_1}+\dfrac{V_{\mathrm{in}}}{L_1}\\[8pt]
\dfrac{y}{C_1}\\[8pt]
-\dfrac{v_o}{L_2}\\[8pt]
\dfrac{i_2}{C_2}-\dfrac{v_o}{RC_2}
\end{bmatrix},
B(x)=
\begin{bmatrix}
\dfrac{v_1}{L_1}\\[8pt]
\dfrac{-y+i_2}{C_1}\\[8pt]
-\dfrac{v_1}{L_2}\\[8pt]
0
\end{bmatrix},
c(x)=x_1.
\]
Since \(c(x)=x_1\),
\[
L_B c(x)=\frac{\partial c}{\partial x}B(x)=\begin{bmatrix}1&0&0&0\end{bmatrix}B(x)=\frac{v_1}{L_1}.
\]
Therefore, \(L_B c(x)\neq 0\) on \(\mathcal{M}_{\mathrm{C}}\), and the system has relative order \(\alpha=1\) on \(\mathcal{M}_{\mathrm{C}}\). By Hirschorn's Theorem 2, the averaged \'Cuk converter is strongly invertible on \(\mathcal{M}_{\mathrm{C}}\).

Moreover,
\[
\dot y = L_A c(x) + u\,L_B c(x)
      = -\frac{v_1}{L_1}+\frac{V_{\mathrm{in}}}{L_1}+\frac{v_1}{L_1}d,
\]
and the inverse input is
\[
d = \frac{L_1\dot y+v_1-V_{\mathrm{in}}}{v_1},
\]
which is well-defined for all \(x\in\mathcal{M}_{\mathrm{C}}\).
\end{proof}

\subsection{Zeta Converter}
\begin{proposition}
For the averaged zeta converter
\begin{align*}
\dot y &= \frac{1}{L_1}\bigl((1-d)v_1+dV_{\mathrm{in}}\bigr),\\
\dot i_2 &= \frac{1}{L_2}\bigl(-dv_1-v_o+dV_{\mathrm{in}}\bigr),\\
\dot v_1 &= \frac{1}{C_1}\bigl((d-1)y+di_2\bigr),\\
\dot v_o &= \frac{1}{C_2}\left(i_2-\frac{v_o}{R}\right),
\end{align*}
with \(x=[y\;\; i_2\;\; v_1\;\; v_o]^\top\) and \(c(x)=y=i_{L_1}\), the system is strongly invertible on
\[
\mathcal{M}_{\mathrm{Z}}=\{x\in\mathbb{R}^4: v_1\neq V_{\mathrm{in}}\}.
\]
\end{proposition}

\begin{proof}
Separating the drift and input vector fields yields
\[
A(x)=
\begin{bmatrix}
\dfrac{v_1}{L_1}\\[8pt]
-\dfrac{v_o}{L_2}\\[8pt]
-\dfrac{y}{C_1}\\[8pt]
\dfrac{i_2}{C_2}-\dfrac{v_o}{RC_2}
\end{bmatrix},
B(x)=
\begin{bmatrix}
\dfrac{V_{\mathrm{in}}-v_1}{L_1}\\[8pt]
\dfrac{V_{\mathrm{in}}-v_1}{L_2}\\[8pt]
\dfrac{y+i_2}{C_1}\\[8pt]
0
\end{bmatrix},
c(x)=x_1.
\]
For this output,
\[
L_Bc(x)
=\begin{bmatrix}1&0&0&0\end{bmatrix}B(x)
=\frac{V_{\mathrm{in}}-v_1}{L_1}.
\]
Thus the only loss of invertibility occurs when \(v_1=V_{\mathrm{in}}\). Away from that set, that is, on \(\mathcal{M}_{\mathrm{Z}}\), the decoupling term remains nonzero and the relative order is again \(\alpha=1\). Hirschorn's Theorem 2 therefore guarantees strong invertibility on \(\mathcal{M}_{\mathrm{Z}}\). Evaluating the output derivative gives
\[
\dot y=L_Ac(x)+uL_Bc(x)
=\frac{v_1}{L_1}+\frac{V_{\mathrm{in}}-v_1}{L_1}d,
\]
so solving for the duty ratio gives
\[
d=\frac{L_1\dot y-v_1}{V_{\mathrm{in}}-v_1}.
\]
The denominator is nonzero precisely on \(\mathcal{M}_{\mathrm{Z}}\), which completes the proof.
\end{proof}

\subsection{SEPIC Converter}
\begin{proposition}
Consider the averaged SEPIC converter
\begin{align*}
\dot y &= -\frac{1-d}{L_1}v_1-\frac{1-d}{L_1}v_o+\frac{V_{\mathrm{in}}}{L_1},\\
\dot i_2 &= \frac{d}{L_2}v_1-\frac{1-d}{L_2}v_o,\\
\dot v_1 &= \frac{1-d}{C_1}y-\frac{d}{C_1}i_2,\\
\dot v_o &= \frac{1-d}{C_2}y+\frac{1-d}{C_2}i_2-\frac{1}{RC_2}v_o,
\end{align*}
with state \(x=[y\;\; i_2\;\; v_1\;\; v_o]^\top\) and output \(c(x)=y=i_{L_1}\). Then strong invertibility holds on
\[
\mathcal{M}_{\mathrm{S}}=\{x\in\mathbb{R}^4: v_1+v_o\neq 0\}.
\]
\end{proposition}

\begin{proof}
For the SEPIC model, the control coefficient in \(\dot y\) suggests the relevant nondegeneracy condition. Express the system in the form \eqref{eq:input_affine} to get
\[
A(x)=
\begin{bmatrix}
-\dfrac{v_1+v_o}{L_1}+\dfrac{V_{\mathrm{in}}}{L_1}\\[8pt]
-\dfrac{v_o}{L_2}\\[8pt]
\dfrac{y}{C_1}\\[8pt]
\dfrac{y+i_2}{C_2}-\dfrac{v_o}{RC_2}
\end{bmatrix},
B(x)=
\begin{bmatrix}
\dfrac{v_1+v_o}{L_1}\\[8pt]
\dfrac{v_1+v_o}{L_2}\\[8pt]
-\dfrac{y+i_2}{C_1}\\[8pt]
-\dfrac{y+i_2}{C_2}
\end{bmatrix}, \\
\]
$c(x)=x_1.$ Therefore,
\[
L_Bc(x)
=\begin{bmatrix}1&0&0&0\end{bmatrix}B(x)
=\frac{v_1+v_o}{L_1}.
\]
The decoupling term is nonzero whenever \(v_1+v_o\neq 0\), namely on \(\mathcal{M}_{\mathrm{S}}\). Hence the output has relative order \(\alpha=1\) throughout that domain, and Hirschorn's Theorem 2 establishes strong invertibility there. Substituting into the input-output relation gives
\[
\dot y=L_Ac(x)+uL_Bc(x)
=-\frac{v_1+v_o}{L_1}+\frac{V_{\mathrm{in}}}{L_1}
+\frac{v_1+v_o}{L_1}d.
\]
Solving for \(d\) yields
\[
d=\frac{L_1\dot y+v_1+v_o-V_{\mathrm{in}}}{v_1+v_o},
\]
which is valid for every \(x\in\mathcal{M}_{\mathrm{S}}\).
\end{proof}

\section{Control Design}

\begin{figure*}[t]
  \centering
\includegraphics{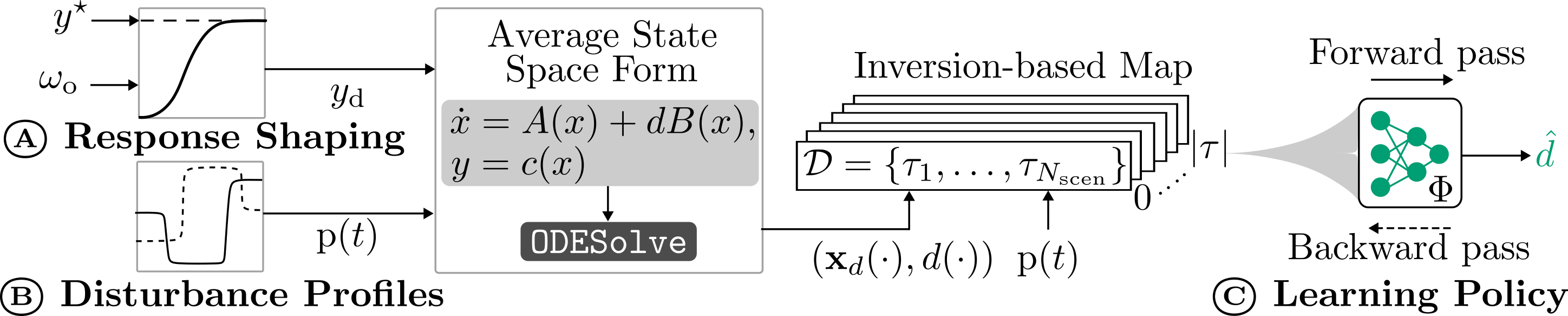}
    \caption{Inversion-based pipeline for learning a desired state-to-duty policy. Given $y^\star$ and $\omega_{\mathrm{o}}$, the inversion dynamics are solved for LHS sampled $(R,V_{\mathrm{in}},x_0)$ to generate trajectory data, $\mathcal{D}=\{\tau_1,\ldots,\tau_{{N_{\text{scen}}}}\}$, which trains a MLP, $\Phi$, to output $\hat{d}$.}
  \label{fig:data_gen_overview}
\end{figure*}

This section shows how data is generated and used to train the inversion-based controller (see Fig.~\ref{fig:data_gen_overview}). 

\subsection{Response Shaping}
We start with a canonical form  as a design template to benefit from classical linear-system theory to impose desirable transient behavior. For the averaged $4^{th}$order converter models we shape the output trajectory using a fourth-order critically damped reference transfer function, with a desired natural frequency, $\omega_0>0$, and amplitude $y^\star$, 
\begin{equation}
    G(s)=\frac{y^\star\omega_0^4}{(s+\omega_0)^4}.
    \label{eq:tf_G}
\end{equation}
 This form is selected to produce a smooth transient response with little overshoot. We apply a unit-step reference, $1/s$, at the input and taking the inverse Laplace transform yields the desired time-domain response 
\begin{align}
    y_d(t)
    &=y^\star\left(1-e^{-\omega_0 t}\left(1+\omega_0 t+\frac{(\omega_0 t)^2}{2!}+\frac{(\omega_0 t)^3}{3!}\right)\right).
    \label{eq:step_response1}
\end{align}

\subsection{Disturbance Profiles}
We model smooth transitions in the operating-point, $p(t)\in\{V_{\mathrm{in}}(t),\,R(t)\}$,  as 
\begin{align*}
s(t) &= \tfrac{1}{2}\left(1+\tanh\!\big(\kappa(t-t_s)\big)\right), \\
p(t) &= (1-s(t))\,p^{-}+s(t)\,p^{+}, \\
\dot p(t) &= \frac{\kappa}{2}\,\mathrm{sech}^2\!\big(\kappa(t-t_s)\big)\,\Delta p.
\end{align*}

Here, $\kappa>0$ determines how abrupt the transition is. As $\kappa$ increases, the profile becomes closer to a step input, whereas smaller values of $\kappa$ produce a more gradual change over a wider time interval. In real-time, changes in $R$ and $V_\text{in}$ are not ideal step functions, and observers are typically used to estimate their values.

\subsection{Learning State-to-duty Policy}

We use the averaged converter model in Section~II to generate desired state and duty-ratio trajectories by inversion. Let the exogenous input vector be $\mathbf{p}(t)=\left[ V_{\mathrm{in}}(t)\;\; R(t)\right]^{\top}$.

For a chosen desired output trajectory \(y_d(\cdot)\), inversion yields a desired state-input pair
\((\mathbf{x}_d(\cdot),d(\cdot))\) that satisfies the converter dynamics and output relation. We train a neural policy to approximate this inversion-generated duty ratio with the static state-parameter map
\begin{equation}
\hat{d}(t)=\pi_{\Phi}\!\bigl(\mathbf{x}(t),\mathbf{p}(t)\bigr),
\end{equation}
where \(\Phi\) denotes the network weights and bias. The policy, $\pi$, uses the current state \(\mathbf{x}(t)\) and exogenous variables \(\mathbf{p}(t)\) to predict the control action \(\hat{d}(t)\approx d(t)\).

We construct a finite dataset of inversion-generated trajectories, $\mathcal{D}=\{\tau_1,\ldots,\tau_{{N_{\text{scen}}}}\},$ where each trajectory has the form
$\tau=
\bigl\{
(\mathbf{x}_0,\mathbf{p}_0,d_0),\ldots,
(\mathbf{x}_{|\tau|},\mathbf{p}_{|\tau|},d_{|\tau|})
\bigr\},$ with $|\tau|=N_t,$ sampling instants
$t_0<t_1<\cdots<t_{|\tau|}.$ Here, \(\mathbf{x}_i\) denotes the state at time \(t_i\), \(\mathbf{p}_i\) denotes the corresponding exogenous input vector, and \(d_i\) denotes the inversion-generated duty ratio. With this dataset \(\mathcal{D}\), we train the multilayer perceptron by minimizing
\begin{equation}
\mathcal{L}(\Phi,\mathcal{D})
=
\sum_{\tau\in\mathcal{D}}
\sum_{(\mathbf{x}_i,\mathbf{p}_i,d_i)\in\tau}
\left\|
\pi_{\Phi}(\mathbf{x}_i,\mathbf{p}_i)-d_i
\right\|^2.
\end{equation}
We minimize \(\mathcal{L}(\Phi,\mathcal{D})\) with gradient-based optimization and then use the trained network to compute the control effort, $\hat{d}$, from the measured state and exogenous variables, as shown in Fig.~\ref{fig:data_gen_overview}.

\begin{figure*}[!h]
\centerline{\includegraphics{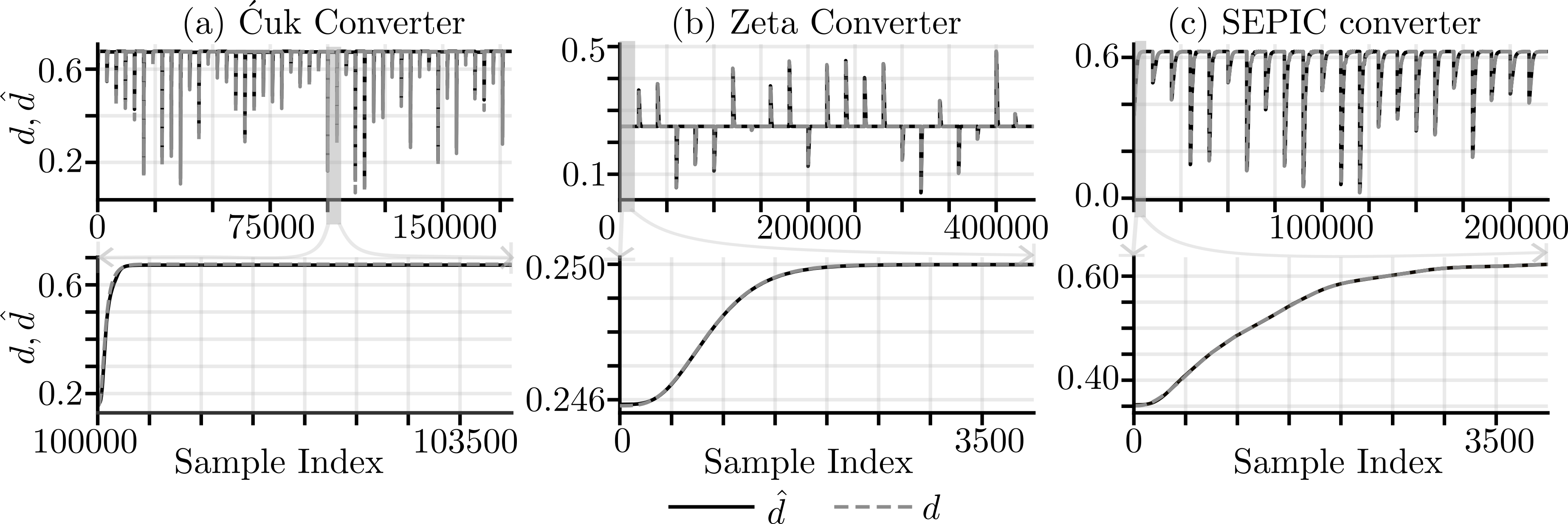}}
\caption{\underline{Validation} - Neural-network validation of duty-cycle prediction for the (a) \'{C}uk, (b) Zeta, and (c) SEPIC converters. The top row compares $\hat{d}$ (solid) and $d$ (dashed) over representative validation trajectories, while the bottom row shows zoomed views of the highlighted regions, demonstrating close agreement of the learned state-to-duty map.}
\label{validation}
\end{figure*}

\begin{figure*}[!h]
\centering
\includegraphics{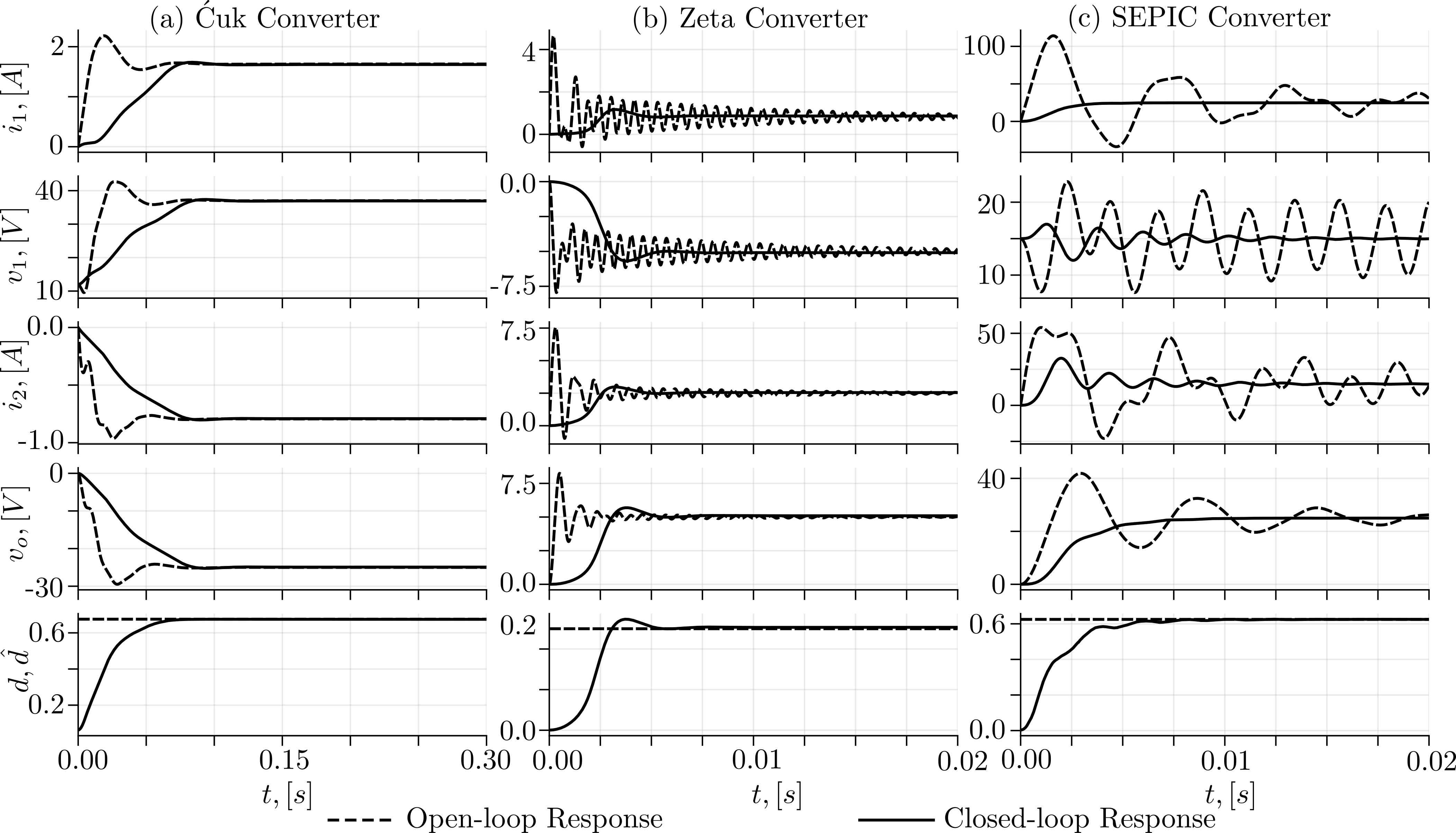}
\caption{\underline{Case \#1} - Open-loop (dashed) and  closed-loop (solid) responses of the states and duty for the three converters. Compared with open-loop, the controller drives the states to their equilibria with faster settling and reduced oscillation.}
\label{Robust_1}
\end{figure*}

\begin{figure*}[!htbp]
\centering
\includegraphics{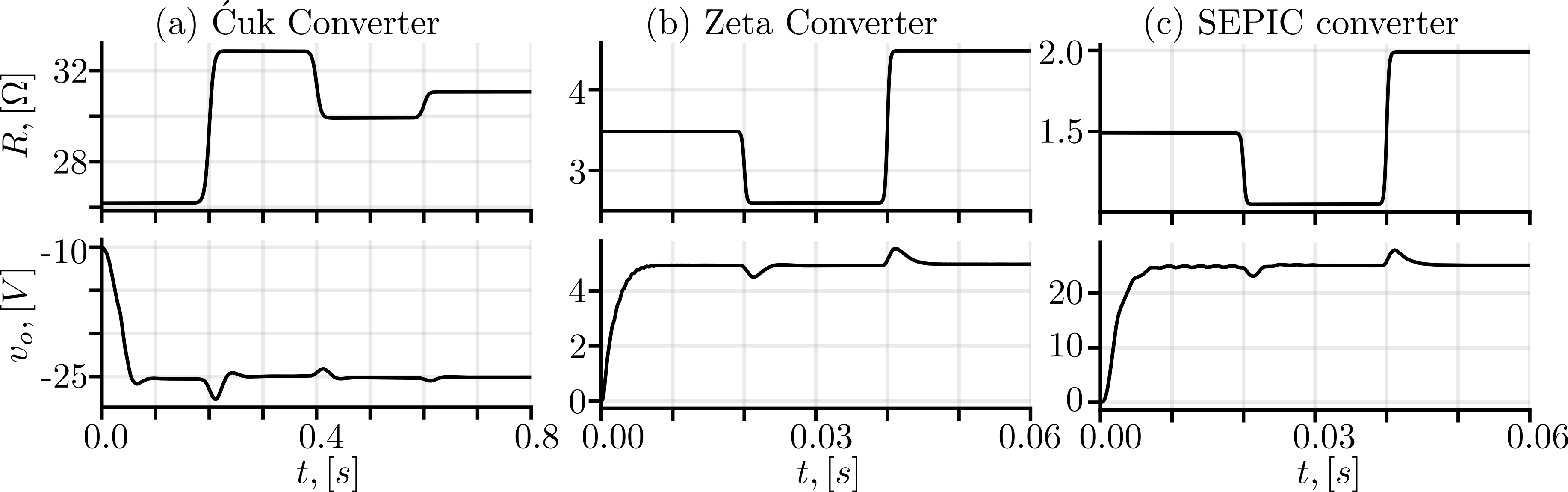}
 \caption{\underline{Case \#2} - Closed-loop robustness to load disturbances for (a) \'{C}uk, (b) Zeta, and (c) SEPIC converters. The top row applies smooth load steps $R(t)$, and the bottom row shows the corresponding output-voltage response $v_o(t)$. 
 Despite abrupt changes in $R$, the neural network controller maintains regulation with small deviations and rapid recovery, demonstrating robustness to load variation across converter topologies a constant nominal input voltage.}
\label{Robust_2}
\end{figure*}

\begin{figure*}[!htbp]
\centering
\includegraphics{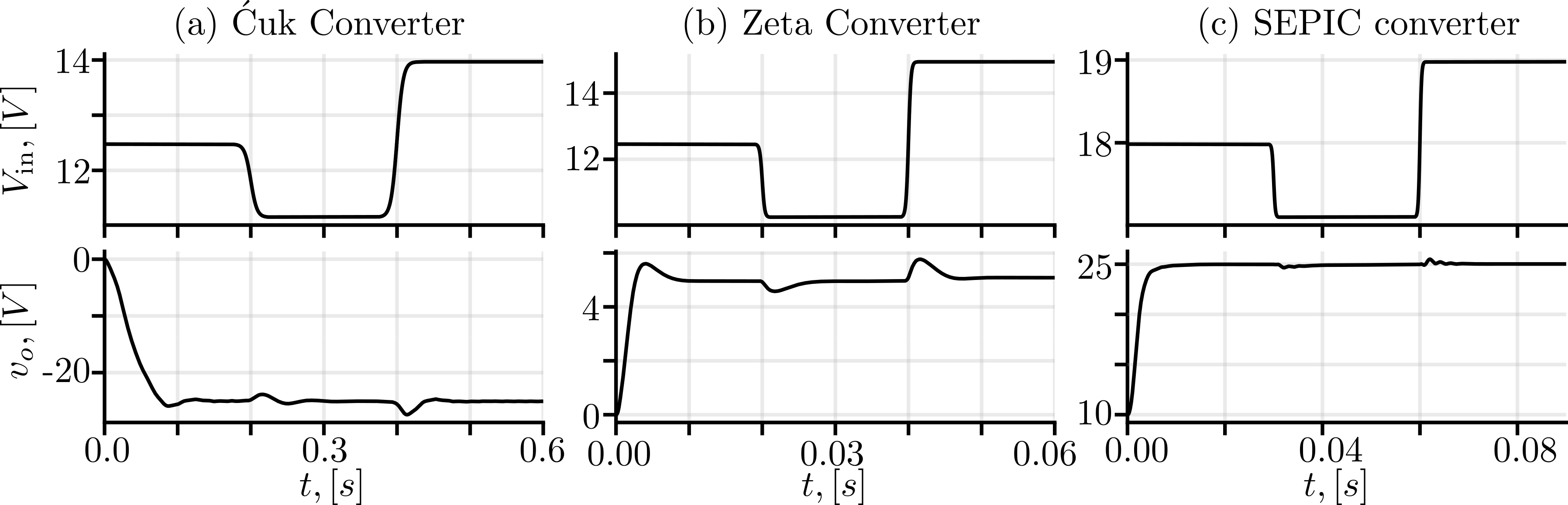}
 \caption{\underline{Case \#3} - Closed-loop robustness to input voltage disturbances for (a) \'{C}uk, (b) Zeta, and (c) SEPIC converters. 
 The top row applies smooth steps of $V_\text{in}(t)$, and the bottom row shows the corresponding output-voltage response $v_o(t)$. Despite abrupt changes in $V_\text{in}(t)$, the neural network controller maintains regulation with small deviations and rapid recovery, demonstrating robustness to input voltage variation across converter topologies with a constant nominal load.}
\label{Robust_3}
\end{figure*}

\begin{figure*}[!htbp]
\centering
\includegraphics{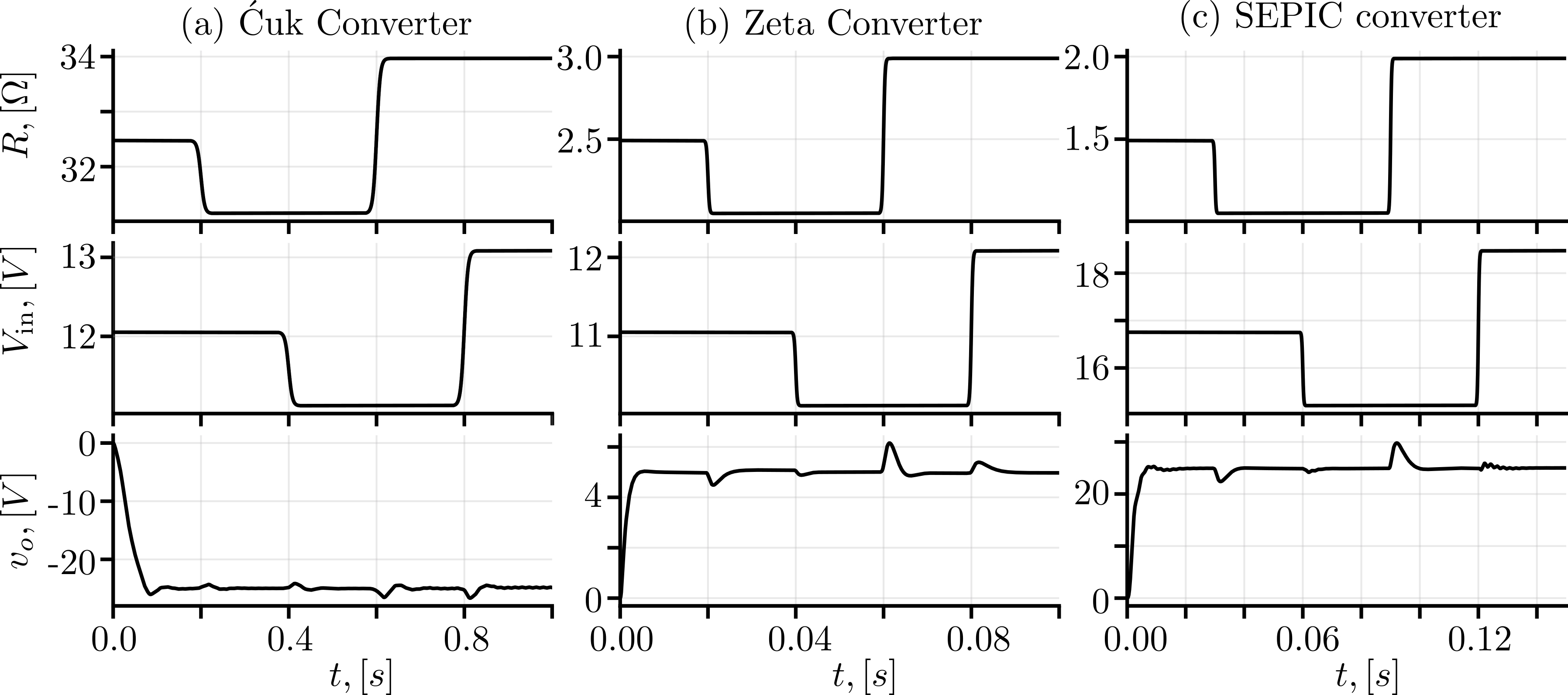}
 \caption{\underline{Case \#4} - Closed-loop robustness to load and input voltage disturbances for (a) \'{C}uk, (b) Zeta, and (c) SEPIC converters. The top two rows apply smooth load and supply voltage steps $(R(t), V_\text{in}(t))$, and the bottom row shows the corresponding output-voltage response $v_o(t)$. 
 Despite large, abrupt changes in $R(t)$ and $V_\text{in}(t)$, the neural network controller maintains regulation with small deviations and rapid recovery, demonstrating robustness to load and input voltage variation across converter topologies.}
\label{Robust_4}
\end{figure*}

 \section{Application to DC-DC Converters} \label{sec:System Description}
\'Cuk, zeta, and SEPIC are dc-dc converters which can step-up or step-down the output voltage. \'Cuk inverts the polarity of the output while zeta and SEPIC are non-inverting. 
\subsection{Bounds}
Table~\ref{tab:eq_ic_relations} summarizes the topology-specific algebraic relations used to define the bounds for the equilibrium point and the initial conditions. In particular, it gives the expressions for computing the equilibrium states and duty ratio from \(R\), \(V_{\text{in}}\), and \(v_o^\star\), as well as the corresponding relations for initializing the states and duty ratio from \(R\), \(V_{\text{in}}\), and \(v_o(0)\).

\begin{table*}[t]
  \centering
  \caption{Equilibrium and initial-condition relations for the \'{C}uk, zeta, and SEPIC converters.}
  \label{tab:eq_ic_relations}
  \renewcommand{\arraystretch}{1.35}
  \footnotesize
  \begin{tabular}{lcccccccccc}
    \toprule
    & \multicolumn{5}{c}{Equilibrium relations} & \multicolumn{5}{c}{Initial-condition relations} \\
    \cmidrule(lr){2-6} \cmidrule(lr){7-11}
    Topology
      & \makecell{$y^\star=i_{L_1}^\star$}
      & \makecell{$i_2^\star$}
      & \makecell{$v_1^\star$}
      & \makecell{$v_o^\star$}
      & \makecell{$d^\star$}
      & \makecell{$y(0)=i_{L_1}(0)$}
      & \makecell{$i_2(0)$}
      & \makecell{$v_1(0)$}
      & \makecell{$v_o(0)$}
      & \makecell{$d(0)$} \\
    \midrule
    \'{C}uk
      & $\dfrac{(v_o^\star)^2}{R\,V_{\text{in}}}$
      & $-\dfrac{v_o^\star}{R}$
      & $V_{\text{in}}+v_o^\star$
      & $-v_o^\star$
      & $\dfrac{v_o^\star}{V_{\text{in}}+v_o^\star}$
      & $\dfrac{(v_o(0))^2}{R\,V_{\text{in}}}$
      & $-\dfrac{v_o(0)}{R}$
      & $V_{\text{in}}+v_o(0)$
      & $-v_o(0)$
      & $\dfrac{v_o(0)}{V_{\text{in}}+v_o(0)}$ \\
    \addlinespace
    Zeta
      & $\dfrac{(v_o^\star)^2}{R\,V_{\text{in}}}$
      & $\dfrac{v_o^\star}{R}$
      & $-v_o^\star$
      & $v_o^\star$
      & $\dfrac{v_o^\star}{V_{\text{in}}+v_o^\star}$
      & $\dfrac{(v_o(0))^2}{R\,V_{\text{in}}}$
      & $\dfrac{v_o(0)}{R}$
      & $-v_o(0)$
      & $v_o(0)$
      & $\dfrac{v_o(0)}{V_{\text{in}}+v_o(0)}$ \\
    \addlinespace
    SEPIC
      & $\dfrac{(v_o^\star)^2}{R\,V_{\text{in}}}$
      & $\dfrac{v_o^\star}{R}$
      & $V_{\text{in}}$
      & $v_o^\star$
      & $\dfrac{v_o^\star}{V_{\text{in}}+v_o^\star}$
      & $\dfrac{(v_o(0))^2}{R\,V_{\text{in}}}$
      & $\dfrac{v_o(0)}{R}$
      & $V_{\text{in}}$
      & $v_o(0)$
      & $\dfrac{v_o(0)}{V_{\text{in}}+v_o(0)}$ \\
    \bottomrule
  \end{tabular}
\end{table*}

\subsection{Trajectory Planning}
For response shaping
of the \'{C}uk, zeta, and SEPIC converters, we choose the desired output as
\begin{equation}
y(t)=i_{L_1}(t),
\end{equation}
and use the common equilibrium relation
\begin{equation}
y^\star = \frac{(v_o^\star)^2}{R\,V_{\text{in}}}.
\end{equation}
To construct a smooth current reference, define the critically damped fourth-order step factor
\begin{equation}
\phi(t)=1-e^{-\omega_0 t}\left(1+\omega_0 t+\frac{(\omega_0 t)^2}{2!}+\frac{(\omega_0 t)^3}{3!}\right).
\end{equation}
Let $Y(t) = \frac{(v_o^\star)^2}{R(t)\,V_{\text{in}}(t)}$ so the current trajectory is
\begin{equation}
y(t)=Y(t)\phi(t),
\end{equation}
\begin{equation}
\dot{y}(t)=\dot{Y}(t)\phi(t)+Y(t)\dot{\phi}(t).
\end{equation}
\textit{\underline{Case \#1.} Constant load resistance and input voltage:}
\begin{equation}
\begin{aligned}
y(t) &= \frac{(v_o^\star)^2}{{R(t)\,V_{\text{in}}(t)}}\,\phi(t), \qquad
\dot{y}(t) = \frac{(v_o^\star)^2}{{R(t)\,V_{\text{in}}(t)}}\,\dot{\phi}(t).
\end{aligned}
\end{equation}
\textit{\underline{Case \#2.} Varying load and constant input voltage:}
\begin{equation}
\begin{aligned}
y(t) &= \frac{(v_o^\star)^2}{R(t)\,V_{\text{in}}}\,\phi(t), \\
\dot{y}(t) &=
-\frac{(v_o^\star)^2\,\dot{R}(t)}{V_{\text{in}}\,R(t)^2}\,\phi(t)
+\frac{(v_o^\star)^2}{R(t)\,V_{\text{in}}}\,\dot{\phi}(t).
\end{aligned}
\end{equation}
\textit{\underline{Case \#3.} Constant load and varying input voltage:}
\begin{equation}
y(t)=\frac{(v_o^\star)^2}{R\,V_{\text{in}}(t)}\,\phi(t), \nonumber
\end{equation} 
\begin{equation}
\dot{y}(t)=
-\frac{(v_o^\star)^2\,\dot{V}_{\text{in}}(t)}{R\,V_{\text{in}}(t)^2}\,\phi(t)
+\frac{(v_o^\star)^2}{R\,V_{\text{in}}(t)}\,\dot{\phi}(t).
\end{equation}

\textit{\underline{Case \#4.} Varying load resistance and input voltage:}
\begin{equation}
y(t)=\frac{(v_o^\star)^2}{R(t)\,V_{\text{in}}(t)}\,\phi(t), \nonumber
\end{equation}
\begin{equation}
\begin{split}
\dot{y}(t)=&
-\frac{(v_o^\star)^2\left(\dot{R}(t)\,V_{\text{in}}(t)+R(t)\,\dot{V}_{\text{in}}(t)\right)}
{\left(R(t)\,V_{\text{in}}(t)\right)^2}\,\phi(t) \\
&+\frac{(v_o^\star)^2}{R(t)\,V_{\text{in}}(t)}\,\dot{\phi}(t).
\end{split}
\end{equation}

\section{Numerical Results}\label{sec:results}


\subsection{Configuration}

We generate a model-based dataset by evaluating the dynamical equations described in Section II, bounds, and trajectories in Section IV over a broad set of operating conditions and transients. For each topology, we specify a reference $\omega_0$, sample operating points, initial conditions, and disturbance profiles (including changes in $V_{\mathrm{in}}$ and $R$), and construct smooth transitions using the shaping parameter $\alpha$. Tables~\ref{tab:nominal_params} lists the nominal parameter values. Tables~\ref{tab:cuk_case_settings}--\ref{tab:sepic_case_settings} list the number of scenarios, $N_\text{scen}$, resolution $N_t$, different disturbance ranges, switching times and simulation horizon $T_\text{sim}$. Sampling across the uncertainty space uses latin hypercube sampling to promote uniform coverage in each dimension \cite{mckay2000comparison}. The resulting data are stored as input output pairs \((z,d)\), where the regressor collects the signals available to the controller and the operating conditions. 
In our simulations we use
\[
z=\begin{bmatrix} i_{L1} & v_1 &  i_{L2} & v_o&V_{\mathrm{in}} & R \end{bmatrix}^{\top}, \qquad d\in[0,1].
\]
The full dataset is split into training and validation subsets using a 70\%/30\% partition. The duty map is represented by a fully connected multilayer perceptron. The configurations used for the reported results are given in Table~\ref{tab:nn_settings_all}, including the activation function, layer widths, batch sizes and epochs. Model training minimizes a mean-squared error loss between the network output and the duty labels with the Adam optimizer. Loss is in the order of $10^{-6}$ for \'Cuk and SEPIC while 0 for zeta.  A constant learning rate of $10^{-6}$ is used for all cases except for Case \#3 and Case \#4 of SEPIC where it was $10^{-7}$. All models were implemented in PyTorch.

\subsection{Inference}
After training, the learned controller is evaluated by simulating the averaged converter in closed loop with fixed-step fourth-order Runge-Kutta integration. Fig.~\ref{validation} compares the predicted duty \(\hat{d}\) with the model-generated \(d\) on held-out \'{C}uk, Zeta, and SEPIC trajectories, showing close agreement of the learned state-to-duty map. Fig.~\ref{Robust_1} then contrasts open-loop and closed-loop responses, where the MLP drives the states to the desired equilibria with faster settling and reduced oscillation. Figs.~\ref{Robust_2}--\ref{Robust_4} examine robustness under four operating scenarios: nominal conditions, smooth load steps in \(R(t)\), smooth input-voltage steps in \(V_{\mathrm{in}}(t)\), and simultaneous variations in both. Across all three converter topologies, the closed-loop responses preserve output-voltage regulation, exhibit only small deviations during disturbances, and recover quickly. These results confirm accurate inference of the inversion-based control law and disturbance rejection across distinct higher-order converter dynamics.

\begin{table*}[t]
  \centering
  \caption{Neural network architecture and training settings for the three converters.}
  \label{tab:nn_settings_all}
  \renewcommand{\arraystretch}{1.15}
  \setlength{\tabcolsep}{5pt}
  \footnotesize
  \begin{tabular}{llcccc}
    \toprule
    Converter & Config 
    & Case \#1
    & Case \#2
    & Case \#3
    & Case \#4 \\
    \midrule

    \multirow{4}{*}{\'{C}uk}
      & Neurons/layer
        & 8/4/2
        & 8/4/2
        & 8/8/8
        & 8/8/8 \\
      & Activations
        & \texttt{sigmoid/sigmoid/sigmoid}
        & \texttt{sigmoid/sigmoid/sigmoid}
        & \texttt{gelu/gelu/sigmoid}
        & \texttt{gelu/gelu/sigmoid} \\
      & Batch size
        & 64
        & 64
        & 32
        & 32 \\
      & Epochs
        & 500
        & 1408
        & 64
        & 64 \\
    \midrule

    \multirow{4}{*}{Zeta}
      & Neurons/layer
        & 128/64/32
        & 128/64/32
        & 128/64/32
        & 64/64/32 \\
      & Activations
        & \texttt{sigmoid/tanh/sigmoid}
        & \texttt{sigmoid/tanh/sigmoid}
        & \texttt{silu/silu/sigmoid}
        & \texttt{gelu/gelu/sigmoid} \\
      & Batch size
        & 32
        & 32
        & 32
        & 8 \\
      & Epochs
        & 17
        & 17
        & 20
        & 5 \\
    \midrule

    \multirow{4}{*}{SEPIC}
      & Neurons/layer
        & 128/128/128
        & 128/128/128
        & 256/128/64
        & 256/128/64 \\
      & Activations
        & \texttt{tanh/tanh/sigmoid}
        & \texttt{tanh/tanh/sigmoid}
        & \texttt{silu/silu/sigmoid}
        & \texttt{silu/silu/sigmoid} \\
      & Batch size
        & 32
        & 32
        & 8
        & 8 \\
      & Epochs
        & 500
        & 56
        & 100
        & 100 \\
    \bottomrule
  \end{tabular}
\end{table*}

\begin{table}[t]
  \centering
  \caption{Nominal parameters.}
  \label{tab:nominal_params}
  \renewcommand{\arraystretch}{1.05}
  \setlength{\tabcolsep}{2.5pt}
  \scriptsize
  \resizebox{\columnwidth}{!}{%
  \begin{tabular}{lccccccccc}
    \toprule
    \makecell{Converter}
    & \makecell{$L_1$\\$[\mathrm{mH}]$}
    & \makecell{$L_2$\\$[\mathrm{mH}]$}
    & \makecell{$C_1$\\$[\mu\mathrm{F}]$}
    & \makecell{$C_2$\\$[\mu\mathrm{F}]$}
    & \makecell{$V_{\mathrm{in}}^{\mathrm{nom}}$\\$[\mathrm{V}]$}
    & \makecell{$V^*$\\$[\mathrm{V}]$}
    & \makecell{$V_o(0)$\\$[\mathrm{V}]$}
    & \makecell{$\kappa$\\}
    & \makecell{$\omega_0$\\} \\
    \midrule
    \'{C}uk & 45   & 45   & 100 & 100  & 12 & -25 & $[-20,0]$ & 120  & 150  \\
    Zeta    & 0.10 & 0.10 & 55  & 200  & 15 & 5   & $[0,15]$  & 3000 & 2000 \\
    SEPIC   & 0.10 & 0.10 & 680 & 2200 & 15 & 25  & $[0,15]$  & 3000 & 2000 \\
    \bottomrule
  \end{tabular}%
  }
\end{table}

\begin{table}[t]
  \centering
  \caption{\'{C}uk converter.}
  \label{tab:cuk_case_settings}
  \renewcommand{\arraystretch}{1.05}
  \setlength{\tabcolsep}{2.5pt}
  \scriptsize
  \resizebox{\columnwidth}{!}{%
  \begin{tabular}{lcccc}
    \toprule
    Setting & Case \#1 & Case \#2 & Case \#3 & Case \#4 \\
    \midrule
    $N_t/N_{\text{scen}}$
      & 4000/300 & 4000/300 & 6000/500 & 6000/500 \\
    $R~[\Omega]$
      & $[25,35]$ & $[25,35]$ & 33 & $[31,34]$ \\
    $V_{\mathrm{in}}~[\mathrm{V}]$
      & 12 & 12 & $[12,13]$ & $[11,14]$ \\
    Switch times [s]
      & none & 0.4 & $[0.2,0.4]$
      & \makecell[c]{$[0.2,0.6]_R$\\$[0.4,0.8]_{V_{\mathrm{in}}}$} \\
    $T_{\mathrm{sim}}~[\mathrm{s}]$
      & 0.8 & 0.8 & 0.6 & 1.0 \\
    \bottomrule
  \end{tabular}%
  }
\end{table}

\begin{table}[t]
  \centering
  \caption{Zeta converter.}
  \label{tab:zeta_case_settings}
  \renewcommand{\arraystretch}{1.05}
  \setlength{\tabcolsep}{2.5pt}
  \scriptsize
  \resizebox{\columnwidth}{!}{%
  \begin{tabular}{lcccc}
    \toprule
    Setting & Case \#1 & Case \#2 & Case \#3 & Case \#4 \\
    \midrule
    $N_t/N_{\text{scen}}$
      & 20000/150 & 20000/150 & 20000/150 & 40000/150 \\
    $R~[\Omega]$
      & $[2.5,4.5]$ & $[2.5,4.5]$ & 2 & $[2,5]$ \\
    $V_{\mathrm{in}}~[\mathrm{V}]$
      & 15 & 15 & $[10,15]$ & $[10,15]$ \\
    Switch times [s]
      & none & $[0.02,0.04]$ & $[0.02,0.04]$
      & \makecell[c]{$[0.02,0.06]_R$\\$[0.04,0.08]_{V_{\mathrm{in}}}$} \\
    $T_{\mathrm{sim}}~[\mathrm{s}]$
      & 0.06 & 0.06 & 0.06 & 0.10 \\
    \bottomrule
  \end{tabular}%
  }
\end{table}

\begin{table}[t]
  \centering
  \caption{SEPIC converter.}
  \label{tab:sepic_case_settings}
  \renewcommand{\arraystretch}{1.05}
  \setlength{\tabcolsep}{2.5pt}
  \scriptsize
  \resizebox{\columnwidth}{!}{%
  \begin{tabular}{lcccc}
    \toprule
    Setting & Case \#1 & Case \#2 & Case \#3 & Case \#4 \\
    \midrule
    $N_t/N_{\text{scen}}$
      & 10000/150 & 20000/150 & 30000/150 & 60000/60 \\
    $R~[\Omega]$
      & $[1,2]$ & $[1,2]$ & 1 & $[1,2]$ \\
    $V_{\mathrm{in}}~[\mathrm{V}]$
      & 15 & 15 & $[15,25]$ & $[15,20]$ \\
    Switch times [s]
      & none & $[0.02,0.04]$ & $[0.03,0.06]$
      & \makecell[c]{$[0.03,0.09]_R$\\$[0.06,0.12]_{V_{\mathrm{in}}}$} \\
    $T_{\mathrm{sim}}~[\mathrm{s}]$
      & 0.02 & 0.06 & 0.09 & 0.15 \\
    \bottomrule
  \end{tabular}%
  }
\end{table}

\section{CONCLUSION}
In this paper, we present a neural network regulator for dc-dc converters. We prove strong invertiblity of averaged dynamical models and use the inversion-based data generation procedure to train a compact multilayer perceptron that maps converter states to the duty ratio. Simulations on the \'Cuk, zeta, and SEPIC converters show close agreement between predicted and model generated duty trajectories, improved closed loop regulation relative to open loop operation, and robust recovery under large load disturbances. This study establishes viability on higher-order averaged converter models over the tested operating envelope. The next phase of this work will incorporate parasitics in the passive and active components, uncertainty-aware models like Gaussian processes, observer-based design, and develop formal closed-loop guarantees under uncertainties.


\section*{ACKNOWLEDGMENT}
This work was supported by the National Science Foundation award number 2409535, ONR N00014-21-1-2164, and a generous gift from Texas Instruments. The authors acknowledge the use of ChatGPT-5.4 for editorial assistance and code formatting support. The authors take full responsibility for the content of this manuscript.



\bibliographystyle{ieeetr}
\bibliography{bibliography_cleanup}
\end{document}